\documentclass[a4paper]{article}
\usepackage[utf8]{inputenc}
\usepackage[sort,comma,square,numbers]{natbib}
\usepackage[bookmarks=false]{hyperref}
\usepackage[table]{xcolor}
\usepackage{amsmath, amssymb, amsthm}
\usepackage{microtype}
\usepackage{subfig}
\usepackage{hyperref}
\usepackage{graphics}
\usepackage{enumerate}
\usepackage{paralist}
\usepackage{booktabs}
\usepackage{tikz}
\usepackage{xspace}
\usepackage{tabulary}
\usepackage{comment}
\usetikzlibrary{arrows,shapes,positioning,calc,decorations.pathreplacing}
\definecolor{linkcol}{rgb}{0,0,0.4} 
\definecolor{citecol}{rgb}{0.5,0,0} 

\usepackage{todonotes}

\usepackage{authblk}

\tikzstyle{vertex}=[circle, draw]
\tikzstyle{bigvertex}=[circle, draw, radius=2]
\tikzstyle{tvertex}=[regular polygon, regular polygon sides=3,inner sep=2pt, draw]
\tikzstyle{rvertex}=[rectangle, inner sep=2pt, draw]
\tikzstyle{edge}=[draw,-]

\makeatletter
\def\NAT@spacechar{~}
\makeatother

\newcommand{\problem}[5]{
\hbox{\vbox{
\begin{quote}
  \label{#5}
  \ifthenelse{\equal{#3}{}}{}{{#3}\ifthenelse{\equal{#4}{}}{}{ ({#4})}}
  \vspace{-2mm}
  \begin{compactdesc}
    \item [Input:] {#1}
    \item [Question:] {#2}
  \end{compactdesc}
\end{quote}
}}
}

\newcommand{\paraproblem}[6]{
\hbox{\vbox{
\begin{quote}
  \label{#6}
  \ifthenelse{\equal{#4}{}}{}{{#4}\ifthenelse{\equal{#5}{}}{}{ ({#5})}}
  \vspace{-2mm}
  \begin{compactdesc}
    \item [Input:] {#1}
    \item [Question:] {#2}
    \item [Parameter:] {#3}
  \end{compactdesc}
\end{quote}
}}
}
\newtheorem{theorem}{Theorem}

\newtheorem{lemma}{Lemma}

\theoremstyle{definition}
\newtheorem{const}{Construction}

\def\OHL{OHL\xspace}
\def\OHLfull{\textsc{Optimal Hub Labeling}\xspace}

\def\VCfull{\textsc{Vertex Cover}\xspace}
\def\VC{\VCfull}

\def\papertitle{Optimal Hub Labeling is NP-complete}
\def\paperauthors{Mathias Weller}
\def\papersubject{Distance Labeling}

\hypersetup
{
bookmarksopen=false,
pdftitle=\papertitle,
pdfauthor=\paperauthors, 
pdfsubject=\papersubject, 
pdfkeywords="",
pdftoolbar=false, 
pdfmenubar=true, 
pdfhighlight=/O, 
colorlinks=true, 
pdfpagemode=UseNone, 
pdfpagelayout=SinglePage, 
pdffitwindow=true, 
linkcolor=linkcol, 
citecolor=citecol, 
urlcolor=linkcol 
}

\title{\papertitle}
\author{\paperauthors%
\thanks{Supported by the CNRS, research project KERNEL.}}%
\affil{LIRMM, Université Montpellier II, France\\ \texttt{mathias.weller@lirmm.fr}}

\begin{document}

\maketitle

\def\uniL{\ensuremath{\tilde{u}_\mathrm{L}}}
\def\uniR{\ensuremath{\tilde{u}_\mathrm{R}}}
\def\VL{\ensuremath{V_\mathrm{L}}}
\def\VR{\ensuremath{V_\mathrm{R}}}

\def\pot#1{\ensuremath{2^{#1}}}

\tikzstyle{uniL}=[]
\tikzstyle{uniR}=[fill=gray!40]
\tikzstyle{hub}=[edge, -latex, thick, dashed]
\newcommand{\hashub}{\raisebox{1mm}{\tikz{\draw[hub] (0,0) -- (.7,0);}}}

\newcommand{\numL}{\gamma}
\newcommand{\numR}{\gamma}

\def\vertgadget#1#2#3#4{%
  \node[vertex] (u#3) at (#1,#2) {};
  \node[vertex] (v#3) at ($(u#3)+(#4:1)$) {} edge (u#3);
  \node[tvertex] (w#3) at ($(v#3)+(#4:1)$) {} edge (v#3);
}
\def\vertgadgetlabeled#1#2#3#4#5{%
  \node[vertex, label=above:$#5_1$] (u#3) at (#1,#2) {};
  \node[vertex, label=above:$#5_2$] (v#3) at ($(u#3)+(#4:1)$) {} edge (u#3);
  \node[tvertex, label=above:$#5_3$] (w#3) at ($(v#3)+(#4:1)$) {} edge (v#3);
}

\begin{abstract}
  \noindent
  Distance labeling is a preprocessing technique introduced by Peleg~[Journal of Graph Theory, 33(3)] 
  to speed up distance queries in large networks.
  Herein, each vertex receives a (short) label and, the distance between two vertices can be inferred from their two labels.
  One such preprocessing problem occurs in the hub labeling algorithm [Abraham et al., SODA'10]: the label of a vertex~$v$ is a set of vertices~$x$ (the ``hubs'') with their distance~$d(x,v)$ to~$v$ and the distance between any two vertices $u$ and~$v$ is the sum of their distances to a common hub.
  The problem of assigning as few such hubs as possible was conjectured to be NP-hard, but no proof was known to date. We give a reduction from the well-known \VCfull problem on graphs to prove that finding an optimal hub labeling is indeed NP-hard.
\end{abstract}

\section{Introduction}\label{sec:intro}

Finding shortest paths quickly is an essential part of many real-world businesses like maps services and programming of mobile GPS navigation devices. In these applications, Dijkstra's algorithm proved much too slow, especially since typical road maps contain millions of vertices and edges. However, since road maps undergo little changes, a preprocessing based approach seems reasonable.
In this spirit, \citet{journal-Pel00} suggested constructing a distributed data structure in advance, allowing later distance queries to be answered in sublinear time. This distributed data structure comprises a (short) label for each vertex such that the distance between two vertices can be inferred from their respective labels.
%
%
Previous work on distance labeling has been focused on finding small labels in polynomial time on a variety of graph classes~%
\cite{journal-BG05
,journal-CDV06
,conf-GP03
,journal-GP03
,journal-GPPR04
,conf-AKM01%
}
or on general graphs~%
\cite{journal-TZ05%
,journal-GPPR04%
} (see also the survey of \citet{journal-GPel03}).

\citet{conf-AFGW10,conf-ADGW11} described the following preprocessing: each vertex~$v$ receives, as a label, a list of vertices~$x$ with their distance~$d(x,v)$ to~$v$ such that, for each two vertices~$u$ and~$v$, there is a vertex (the ``hub'') on a shortest $u$-$v$-path that is in both, the label of~$u$ and the label of~$v$. Once all labels have been computed, the distance between~$u$ and~$v$ is the minimum over all vertices~$x$ occurring in the intersection of the labels of~$u$ and~$v$ of the sum of $d(x,u)$ and~$d(x,v)$.
In this work, we focus on the computational complexity of this preprocessing. In particular, when the goal is to minimize the overhead storage needed for the labeling, it involves solving (the optimization variant of) the following problem.

\problem%
{A graph~$G=(V,E)$ and an integer~$k$.}
{Is there an assignment~$\ell:V\to\pot{V}$ such that~$\sum_{v\in V}|\ell(v)|\leq k$ and for all $u,v\in V$, some vertex of some shortest $u$-$v$-path in~$G$ is in $\ell(u)\cap\ell(v)$ (specifically, we allow~$u=v$, thereby requiring~$\ell$ to be reflexive)?}
{\OHLfull}{\OHL}{def:OHL}

\looseness=-1
\noindent While \citet{conf-AFGW10} conjectured that \OHL is NP-hard, no hardness-proof was known to date. We present a polynomial-time reduction of the well-known \VCfull problem to \OHL, thereby demonstrating its NP-hardness. Since an optimal hub-labeling can be verified by computing the lengths of all shortest paths and comparing them to the distances of each vertex to its hubs, \OHL is also contained in NP, implying NP-completeness for the problem.
Our work falls in line with \citet{report-BCKKW10} who proved various exact preprocessing problems NP-hard that were designed to speed up routing or distance queries. To motivate, they point out that the majority of the known results in this area are heuristic and only few results about exact computational complexity are known. \citet{journal-CHKZ03} provide an exception to this observation, proving that a preprocessing variant called \emph{2-hop cover} that is quite similar to hub labeling is NP-hard. Unfortunately, we were unable to reuse their prove to prove \OHL NP-hard.


\subsection{Preliminaries}
\looseness=-1
Let~$G=(V,E)$ be a graph and let~$\ell:V\to\pot{V}$ be a mapping. We say that~$\ell$ \emph{covers} a shortest path~$p$ between two vertices~$u,v\in V$ \emph{with} a vertex~$x$ if~$x$ is on~$p$ and~$x\in\ell(u)\cap\ell(v)$. If~$x$ is either clear from context or unknown, then we simply say $p$ is \emph{covered by}~$\ell$. If, for each $u,v\in V$ some shortest $u$-$v$-path in~$G$ is covered by~$\ell$ (including degenerate cases where~$u=v$), then we call~$\ell$ a \emph{hub-labeling} of~$G$. When clear from context, we drop the suffix ``of~$G$''.
Slightly abusing notation, we identify~$\ell$ with the set of pairs~$(x,y)$ with~$y\in\ell(x)$. Thus, the \emph{size} of~$\ell$ is~$|\{(x,y)\mid y\in\ell(x)\}|$, denoting the total number of assignments of~$\ell$. Finally, $\ell$ is said to be \emph{optimal} if no hub labeling of~$G$ is strictly smaller than~$\ell$.
We use~$\ell^{-1}$ to denote~$\{(x,y)\mid (y,x)\in\ell\}$. For brevity, we abbreviate size-2 sets~$\{u,v\}$ to~$uv$ and sets~$\{1,2,\ldots,i\}$ to~$[i]$.

\section{Detailed Reduction}
In this section, we give the reduction of \VCfull to \OHLfull, explain details, and prove its correctness. To this end, we establish a general form that optimal solutions for the created instance of \OHL can be assumed to have. We show that the way in which the shortest paths are covered corresponds to a vertex cover of the input graph.

\begin{const}\label{const:reduction}
  Let~$(G'=(V',E'),k')$ be an instance of \VC and let~$\gamma:=8|V'|+3|E'|+k'+2$. We construct an instance~$(G=(V,E),k)$ of \OHL as follows. 
  \begin{inparaenum}[1.]
    \item Add $\gamma$ new isolated vertices~$w_i$,
    \item add a new universal vertex~$w$,
    \item rename each~$v\in V'$ to~$v_1$,
    \item add a private neighbor~$v_2$ to each~$v_1$, and
    \item add a private neighbor~$v_3$ to each~$v_2$.
  \end{inparaenum}
  More formally,
  \begin{align*}
    W  & := \{w_1,w_2,\ldots,w_\gamma\}\\
    V  & := \{w\}\cup W\cup\{v_1,v_2,v_3 \mid v\in V'\}\\
    E  & := \bigcup_{v\in V'}\{wv_1,v_1v_2,v_2v_3\}\cup\bigcup_{x\in W}\{wx\}\cup\bigcup_{uv\in E'}\{u_1v_1\}
  \end{align*}
  Finally, let $k := 3\gamma-1$.
  An example of the construction is sketched in \autoref{fig:reduction}.
\end{const}

\def\decorate#1#2{%
  \pgfmathsetmacro{\decoa}{#2+25}
  \pgfmathsetmacro{\decob}{#2-20}
  \pgfmathsetmacro{\decoc}{#2-50}
  \path[edge] (#1) -- ++(\decoa:0.4);
  \path[edge] (#1) -- ++(\decob:0.5);
  \path[edge] (#1) -- ++(\decoc:0.4);
}
\def\vertgadgettrue#1#2#3#4{%
  \vertgadget{#1}{#2}{#3}{#4}
  \decorate{u#3}{180-#4}
  \draw (w#3) edge[hub, bend right] (v#3);
  \draw (w#3) edge[hub, bend left] (u#3);
  \draw (v#3) edge[hub, bend right] (u#3);
}

\def\vertgadgetfalse#1#2#3#4{%
  \vertgadget{#1}{#2}{#3}{#4}
  \decorate{u#3}{180-#4}
  \draw (u#3) edge[hub, bend left] (v#3);
  \draw (w#3) edge[hub, bend right] (v#3);
}

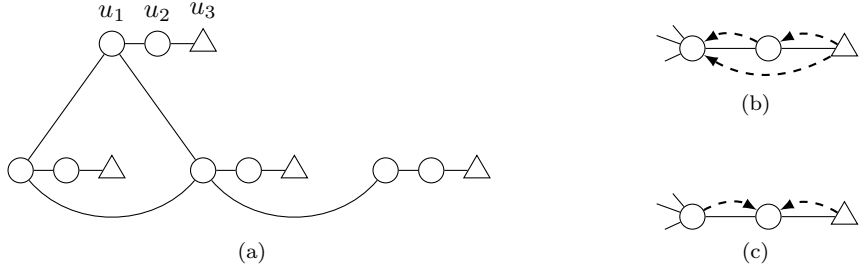
\begin{figure}[t]
  \centering
  \subfloat[]{
    \begin{tikzpicture}[scale=.6, node distance=6mm, bend angle=50]
      \vertgadget{0}{0}{1}{0}
      \vertgadgetlabeled{2}{2.8}{2}{0}{u}
      \vertgadget{4}{0}{3}{0}
      \vertgadget{8}{0}{4}{0}

      \draw[edge] (u1) -- (u2) -- (u3);
      \draw[edge] (u1) edge [bend right] (u3);
      \draw[edge] (u3) edge [bend right] (u4);
    \end{tikzpicture}
    \label{fig:reduction}
  }
  \hspace{15mm}
  \begin{tabular}[b]{c}
  \subfloat[]{
    \begin{tikzpicture}[node distance=6mm]
      \vertgadgettrue{0}{0}{1}{0}
    \end{tikzpicture}
    \label{fig:vg true}
  }\\[10mm]
  \subfloat[]{
    \begin{tikzpicture}[node distance=6mm]
      \vertgadgetfalse{0}{0}{1}{0}
    \end{tikzpicture}
    \label{fig:vg false}
  }
  \end{tabular}

  \caption[Example of Construction]{\autoref{fig:reduction} shows the result of \autoref{const:reduction} for $G'=$
    \raisebox{-2pt}{
      \resizebox{!}{12pt}{
        \begin{tikzpicture}[scale=.6]
          \node[vertex] (a) at (7,3) {};
          \node[vertex] (b) at (7.5,3.7) {} edge (a);
          \node[vertex] (c) at (8,3) {} edge (a) edge (b);
          \node[vertex] (d) at (9,3) {} edge (c);
        \end{tikzpicture}
      }
    }\hspace{-1.2mm}.
  Vertices $w$ and~$w_i$ are omitted. The closed neighborhood of any triangular vertex ($v_3$ for each~$v\in V'$) is a proper subset of the closed neighborhood of the adjacent round vertex.
  \autoref{fig:vg true} and \ref{fig:vg false} show possible partial solutions $\ell_v$ (dashed arcs) for the ``vertex gadget'' of~$v\in V'$: \autoref{fig:vg true} represents choosing~$v$ for the vertex cover, \autoref{fig:vg false} represents not choosing~$v$.}
  \label{fig:vertex gadget}
\end{figure}

\looseness=-1
Given a hub labeling~$\ell$ for~$G$, for each~$v\in V'$, we define~$\ell_v:=\ell\cap\{(v_i,v_j)\mid i,j\in[3] \wedge i\ne j\}$ to denote the set of non-reflexive assignments of~$\ell$ in the vertex gadget of~$v$ and, for each~$uv\in E'$, we define~$\ell_{uv}:=\ell\cap\{(u_i,v_j),(v_j,u_i)\mid i,j\in[3]\}$ to denote the set of assignments of~$\ell$ between the vertex gadgets of~$u$ and~$v$.
Note that, for each~$v\in V'$, $|\ell_v|\geq 2$ since a single assignment cannot cover the shortest paths~$(v_1,v_2)$ and~$(v_2,v_3)$ (see \autoref{fig:vg true} and \ref{fig:vg false}). Likewise, for all~$uv\in E'$, $|\ell_{uv}|\geq 3$ since, for each~$i\in [3]$ the unique shortest $u_i$-$v_i$-path in~$G$ requires a different assignment between the vertex gadgets of~$u$ and~$v$ (see \autoref{fig:edge gadget}).

\begin{lemma}\label{obs:normalized}
  Let~$(G,k)$ be a yes-instance of \OHL constructed by \autoref{const:reduction}.
  Then, there is an optimal hub labeling~$\ell$ for~$G$ such that
  \begin{compactenum}[(1)]
    \item\label{obs:subsets} 
      for all~$u,v\in V$ with~$N_G[u]\subset N_G[v]$, $u\notin\ell(v)$,
    \item\label{obs:universals}
      $\ell^{-1}(w)=V$ and
      $\ell^{-1}(x) = \{x\}$ for all~$x\in W$,
    \item\label{obs:vertex internal}
      for all~$v\in V'$, $v_1\in\ell(v_2)\Rightarrow |\ell_v|>2$, and
    \item\label{obs:VC}
      for all~$uv\in E'$, $u_1\notin\ell(u_2) \wedge v_1\notin\ell(v_2) \Rightarrow |\ell_{uv}|>3$.
  \end{compactenum}
  We call such a hub labeling \emph{normalized}.
\end{lemma}
%
%
%
\begin{proof} 
  Let~$\ell$ be an optimal hub labeling. For each of the properties in \autoref{obs:normalized}, we suppose that~$\ell$ has all previous properties. Then, we transform~$\ell$, in each step achieving one of the properties in \autoref{obs:normalized} without destroying any of the previous properties, or increasing the size of~$\ell$. Thus, the result is a normalized optimal hub labeling.
  \begin{compactdesc}
    \item[\eqref{obs:subsets}:]
      Let~$u,v\in V$ with~$N[u]\subset N[v]$, let~$z\in V\setminus\{u\}$, and assume that~$u$ is in a shortest $v$-$z$-path~$p$ of~$G$. Then, $p=(v,u,x,\ldots,z)$. However, since~$x\in N[u]$ implies~$x\in N[v]$, we know that~$p$ is not a shortest path. Thus, removing~$u$ from~$\ell(v)$ uncovers only the path~$(u,v)$, which can then be covered by adding~$v$ to~$\ell(u)$. Clearly, this modification does not destroy \eqref{obs:subsets} for any pair of vertices.

    \item[\eqref{obs:universals}:]
      By \eqref{obs:subsets}, we know that~$W\subseteq\ell^{-1}(w)$.
      Furthermore, each~$x\in W$ occurs only in shortest paths that end with~$x$. Since all these paths contain~$w$, we can replace $x$ with~$w$ in each assignment of~$\ell$ except for the reflexive~$(x,x)$ without loosing \eqref{obs:subsets}.
      Now, since $\sum_{x\in W}3= 3\gamma>k$, there is some~$x\in W$ with~$|\ell(x)|<3$, implying~$\ell(x)=\{x,w\}$. Then, since~$\ell^{-1}(x)=\{x\}$ and all vertices in~$V$ have a shortest path to~$x$, we conclude~$\ell^{-1}(w)=V$.



    \item[\eqref{obs:vertex internal}]
      Let~$v_1\in\ell(v_2)$ for some~$v\in V'$. 
      Then, to cover the shortest paths
      $(v_2,v_3)$ and $(v_1,v_2,v_3)$, 
      $\ell_v$ contains two different assignments, each of which differs from~$(v_2,v_1)$ (see \autoref{fig:vg true}).
      Thus, we conclude~$|\ell_v|>2$.

    \item[\eqref{obs:VC}]
      Let~$uv\in E'$ such that~$u_1\notin\ell(u_2)$ and~$v_1\notin\ell(v_2)$.
      Then, for each of the following shortest paths, $\ell_{uv}$ contains a different assignment (see \autoref{fig:eg false}): $(u_1,v_1)$, $(u_2,u_1,v_1)$, $(u_1,v_1,v_2)$, $(u_2,u_1,v_1,v_2)$. 
      Thus, we conclude~$|\ell_{uv}|>3$.\qedhere
  \end{compactdesc}
\end{proof}

\def\Cc{\cellcolor{gray!20}}
\def\bw{\boldsymbol{w}}


Intuitively speaking, covering $(v_1,v_2)$ with~$v_1$ for some~$v\in V'$ induces more cost (by \eqref{obs:vertex internal}) and will correspond to choosing~$v$ into a vertex cover of~$G'$. However, by \eqref{obs:VC}, choosing neither~$u$ nor~$v$ for any~$uv\in E'$ also induces more cost, allowing us to just take one of $uv$ into the vertex cover instead.

\begin{figure}[t]
  \centering
  \subfloat[]{
    \begin{tikzpicture}[node distance=6mm]
      \vertgadgettrue{0}{3}{1}{0}
      \vertgadgettrue{0}{0}{2}{0}
      \draw (u1) edge (u2);

      \draw (u2) edge[hub,bend left] (u1);
      \draw (v2) edge[hub] (u1);
      \draw (w2) edge[hub] (u1);
    \end{tikzpicture}
    \label{fig:eg both}
  }
  \hspace{10mm}
  \subfloat[]{
    \begin{tikzpicture}[node distance=6mm]
      \vertgadgettrue{0}{3}{1}{0}
      \vertgadgetfalse{0}{0}{2}{0}
      \draw (u1) edge (u2);
 
      \draw (u2) edge[hub,bend left] (u1);
      \draw (v2) edge[hub] (u1);
      \draw (w2) edge[hub] (u1);
    \end{tikzpicture}
    \label{fig:eg true}
  }
  \hspace{10mm}
  \subfloat[]{
    \begin{tikzpicture}[node distance=6mm]
      \vertgadgetfalse{0}{3}{1}{0}
      \vertgadgetfalse{0}{0}{2}{0}
      \draw (u1) edge (u2);

      \draw (u2) edge[hub,bend left] (u1);
      \draw (v2) edge[hub] (v1);
      \draw (v2) edge[hub] (u1);
      \draw (v1) edge[hub] (u2);
    \end{tikzpicture}
    \label{fig:eg false}
  }

  \caption{\looseness=-1 Illustration of a gadget corresponding to some edge~$uv\in E'$ ($u$ on top, $v$ below) with a possible partial solution~$\ell_{uv}$ (dashed arcs). The figures correspond to different vertex covers of~$G$: choosing both~$u$ and~$v$ (\autoref{fig:eg both}) or choosing~$u$ but not~$v$ (\autoref{fig:eg true}). Choosing neither~$u$ nor~$v$ (\autoref{fig:eg false}) causes additional assignments in~$\ell_{uv}$, implying~$|\ell_{uv}|>3$ (see \autoref{obs:normalized}\eqref{obs:VC}).}
  \label{fig:edge gadget}
\end{figure}
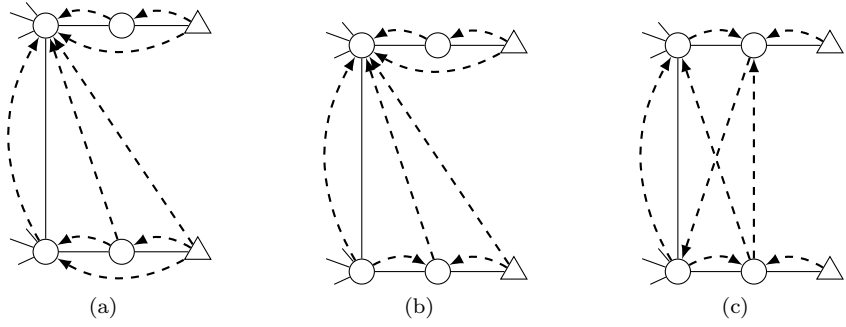

\def\rawl{\ensuremath{\ell_{\text{raw}}}}

\begin{lemma}\label{lem:VC is HL}
  Let~$X$ be a size-$k'$ vertex-cover of~$G'$. Then, there is a size-$k$ hub labeling for~$G$.
\end{lemma}
\begin{proof}
  Let~$f:E'\to V'$ be a function mapping each~$uv\in E'$ to some vertex in~$uv\cap X$. We will use $f$ to ``break ties'' between two vertices in~$X$ that are adjacent in~$G'$. Let~$\ell',\ell''$ be assignments such that (see \autoref{fig:vg true}, \ref{fig:vg false}, and \ref{fig:edge gadget})
  \begin{compactenum}
    \item for all~$v\in X$ and~$i\in[3]$, set~$\ell'(v_i)=\{v_j\mid 1\leq j<i\}$,
    \item for all~$v\in V'\setminus X$, set~$\ell'(v_1)=\ell'(v_3)=\{v_2\}$ and~$\ell'(v_2)=\emptyset$, and
    \item for all~$uv\in E'$ with~$f(uv)=v$ and all~$i\in[3]$, set~$\ell''(u_i)=\{v_1\}$.
  \end{compactenum}
  Then, let~$\ell:=\ell'\cup\ell''\cup\{(x,x),(x,w)\mid x\in V\}$.
  Since~$|X|\leq k'$, we have~
  \begin{align*}
    |\ell'\cup\ell''| & =  3|X|+2|V'\setminus X|+3|E'|\leq 2|V'|+3|E'|+k'\text{, implying}\\
    |\ell| & \leq 2|V'|+3|E'|+k'+2(3|V'|+\gamma)+1= 
        3\gamma-1
  \end{align*}
  In the following, we show that all shortest paths of~$G$ are covered by~$\ell$.
  First, all shortest paths of length 0 are covered by~$\ell$.
  Second, for each~$v\in V'\setminus X$, the shortest paths~$(v_1,v_2)$, $(v_2,v_3)$, and~$(v_1,v_2,v_3)$ are covered with~$v_2$ and, for each~$v\in X$, the shortest paths~$(v_1,v_2)$ and~$(v_1,v_2,v_3)$ are covered with~$v_1$ and the shortest path~$(v_2,v_3)$ is covered with~$v_2$.
  Third, for each~$uv\in E'$ with $f(uv)=v$ and all~$i,j\in[3]$, the unique shortest $u_i$-$v_j$-path in~$G'$ is covered with~$v_1$.
  Since for each~$uv\notin E'$ and all~$i,j\in[3]$, there is a shortest $u_i$-$v_j$-path containing~$w$ and all shortest paths containing~$w$ are covered with~$w$, we conclude that~$\ell$ is indeed a hub labeling for~$G$ and its cost is at most~$3\gamma-1=k$.
\end{proof}

\begin{lemma}\label{lem:HL is VC}
  Let~$\ell$ be a normalized optimal hub labeling for~$G$ and let~$|\ell|\leq k$.
  Then, there is a size-$k'$ vertex cover for~$G'$.
\end{lemma}
\begin{proof}
  First, let~$\ell_1:=\{(x,x),(x,w)\mid x\in V\}$ and note that, by \autoref{obs:normalized}\eqref{obs:universals}, $\ell_1\subseteq\ell$.
  Then, $|\ell\setminus\ell_1|\leq 3\gamma-1 - 2\cdot(3|V'|+\gamma)-1=\gamma-6|V'|-2=2|V'|+3|E'|+k'$.

  To break ties, let~$f:E'\to V'$ be an arbitrary function with~$f(uv)\in uv$ for each~$uv\in E'$.
  Let~$X_1:=\{v \in V'\mid |\ell_v|>2\}$ and~$X_2:=\{f(uv) \in V'\mid |\ell_{uv}|>3\}$. Since, for each~$v\in V'$, we have~$|\ell_v|\geq 2$ and~$\ell_v\cap\ell_1=\emptyset$ and, for each~$uv\in E'$, we have~$|\ell_{uv}|\geq 3$, $\ell_{uv}\cap\ell_1=\emptyset$, and~$(\ell_u\cup\ell_v)\cap\ell_{uv}=\emptyset$, we conclude
  \begin{align*}
    |\ell\setminus\ell_1| & \geq \sum_{v\in V'}|\ell_v|+\sum_{uv\in E'}|\ell_{uv}|\\
                          & \geq 3|X_1|+2(|V'|-|X_1|)+4|X_2|+3(|E'|-|X_2|)\\
                          & =2|V'|+3|E'|+|X_1|+|X_2|,
  \end{align*}
  directly implying~$|X_1\cup X_2|\leq k'$.

  To see that~$X_1\cup X_2$ is a vertex cover for~$G'$, assume there is some~$uv\in E'$ such that~$uv\cap(X_1\cup X_2)=\emptyset$.
  Then, since~$u,v\notin X_1$, \autoref{obs:normalized}\eqref{obs:vertex internal} implies~$u_1\notin\ell(u_2)$ and~$v_1\notin\ell(v_2)$.
  But, since~$u,v\notin X_2$, \autoref{obs:normalized}\eqref{obs:VC} implies~$u_1\in\ell(u_2)$ or~$v_1\in\ell(v_2)$, a contradiction.
\end{proof}

With \autoref{obs:normalized}, \autoref{lem:VC is HL} and \autoref{lem:HL is VC} imply the main theorem which, since \VCfull is NP-hard on planar graphs, holds also for apex graphs.
\begin{theorem}
  \OHLfull is NP-complete, even on apex graphs.
\end{theorem}


\let\oldthebibliography=\thebibliography
\let\endoldthebibliography=\endthebibliography
\renewenvironment{thebibliography}[1]{%
  \begin{oldthebibliography}{#1}%
    \setlength{\parskip}{0ex}%
    \setlength{\itemsep}{0ex}%
    \footnotesize
}{
    \end{oldthebibliography}%
}
\bibliographystyle{abbrvnat}
\bibliography{hl}

\end{document}